\documentclass[a4paper,11pt]{article}

\usepackage{amsmath,amsthm,enumerate,color}
\usepackage{amssymb,graphicx}
\usepackage{accents}
\usepackage{mathdots}
\usepackage[tableposition=above]{caption}
\usepackage{subcaption}
\usepackage{bbm,hyperref}
\hypersetup{
	colorlinks=true,
    linkcolor={red!50!black},
    citecolor={blue!50!black},
    urlcolor={blue!80!black},
	bookmarksopen=true,
	bookmarksnumbered,
	bookmarksopenlevel=2,
	bookmarksdepth=3
}
\usepackage[capitalize]{cleveref}
\usepackage{float}
\usepackage{nicefrac}
\usepackage[utf8]{inputenc}
\usepackage{centernot}
\usepackage[]{todonotes} 
\usepackage{multirow}
\usepackage{multicol}
\usepackage{etoolbox}
\usepackage{mathtools}
\usepackage{comment}
\usepackage[normalem]{ulem}
\usepackage[ruled,vlined,linesnumbered]{algorithm2e}
\crefalias{algocf}{algorithm}\crefalias{algocfline}{line}\crefalias{AlgoLine}{line}
\providetoggle{long}
\settoggle{long}{true}
\usepackage{tikz}
\usetikzlibrary{arrows}
\usetikzlibrary{shapes}
\tikzstyle{every path}=[thick]

\long\def\Omit#1{}
\def\diss{{\delta}}
\def\totdiss{\textit{diss}}

\usepackage{microtype}

\newtheorem{theorem}{Theorem}
\newtheorem{lemma}[theorem]{Lemma}
\newtheorem{observation}[theorem]{Observation}

\newcounter{claimCount}
\AtBeginEnvironment{proof}{\setcounter{claimCount}{0}}

\theoremstyle{remark}
\newtheorem*{ccproof}{Proof}

\usepackage{tabularx, environ}

\makeatletter
    
\newcolumntype{\expand}{}
\long\@namedef{NC@rewrite@\string\expand}{\expandafter\NC@find}

\NewEnviron{fancyproblem}[2][]{	\def\problem@arg{#1}	\def\problem@framed{framed}	\def\problem@lined{lined}	\def\problem@doublelined{doublelined}	\ifx\problem@arg\@empty	\def\problem@hline{}	\else	\ifx\problem@arg\problem@doublelined	\def\problem@hline{\hline\hline}	\else	\def\problem@hline{\hline}	\fi	\fi	\ifx\problem@arg\problem@framed	\def\problem@tablelayout{|>{\bfseries}lX|c}	\def\problem@title{\multicolumn{2}{|l|}{			\raisebox{-\fboxsep}{\textsc{\large #2}}	}}	\else
	\def\problem@tablelayout{>{\bfseries}lXc}	\def\problem@title{\multicolumn{2}{l}{			\raisebox{-\fboxsep}{\textsc{\large #2}}	}}	\fi	\bigskip\par\noindent	 \begin{tabularx}{\textwidth}{\expand\problem@tablelayout}		\problem@hline		\problem@title\\[2\fboxsep]		\BODY\\\problem@hline	\end{tabularx}	\medskip\par}
\makeatother
\colorlet{darkgreen}{green!60!black}
 	
\def\real{\hbox{\rm\vrule\kern-1pt R}}
\def\nat{\hbox{\rm\vrule\kern-1pt N}}

\newcommand{\PP}{\textsf{P}}
\newcommand{\NP}{\textsf{NP}}
\newcommand{\MINMAX}{\textsc{CG Min-Max Dissatisfaction}}

\newcommand{\pred}{\textnormal{pred}}
\newcommand{\succc}{\textnormal{succ}}

\newcommand{\probdef}[6][Question]{
\hbox{\vbox{
  \medskip
  \ifthenelse{\equal{#5}{}}{}{\label{#6}}
  \noindent\ifthenelse{\equal{#4}{}}{}{\textsc{#4}\ifthenelse{\equal{#5}{}}{}{ \sloppy\mbox{(\textsc{#5})}}}
  \par
  \smallskip
  \noindent\begin{tabularx}{\textwidth}{@{}l X}
    \textbf{Input:}		& {#2}\\
    \textbf{#1:}		& {#3}
  \end{tabularx}
  \medskip
}}
}

\newcommand{\picturepathsamelength}{
 \begin{tikzpicture}
 \begin{scope}[scale=0.5]
\node[fill=white,draw,circle,inner sep=0pt,minimum size=2mm,label=left:{$1$}] (u1) at (0,0)   {};
\node[fill=white,draw,circle,inner sep=0pt,minimum size=2mm,label=left:{$2$}] (u2) at (0,-1)   {};
\node[fill=none,draw=none,circle,inner sep=0pt,minimum size=0mm,label=left:{$\vdots$}] (u3) at (0,-2)   {};
\node[fill=white,draw,circle,inner sep=0pt,minimum size=2mm,label=left:{$\frac{k-1}{2}$}] (u4) at (0,-3)   {};
\node[fill=white,draw,circle,inner sep=0pt,minimum size=2mm,label=left:{$\frac{k+1}{2}$}] (u5) at (0,-4)   {};
\node[fill=white,draw,circle,inner sep=0pt,minimum size=2mm,label=left:{$\frac{k+3}{2}$}] (u6) at (0,-5)   {};
\node[fill=none,draw=none,circle,inner sep=0pt,minimum size=0mm,label=left:{$\vdots$}] (u7) at (0,-6)   {};
\node[fill=white,draw,circle,inner sep=0pt,minimum size=2mm,label=left:{$k-1$}] (u8) at (0,-7)   {};
\node[fill=white,draw,circle,inner sep=0pt,minimum size=2mm,label=left:{$k$}] (u9) at (0,-8)   {};
 
\draw(u1) -- (u2);
\draw[dashed](u2) -- (u3);
\draw[dashed](u3) -- (u4);
\draw(u4) -- (u5);
\draw(u5) -- (u6);
\draw[dashed](u6) -- (u7);
\draw[dashed](u7) -- (u8);
\draw(u8) -- (u9);
\end{scope}
 
\begin{scope}[scale=0.5,xshift=3cm]
\node[fill=white,draw,circle,inner sep=0pt,minimum size=2mm,label=left:{$k-1$}] (u1) at (0,0)   {};
\node[fill=white,draw,circle,inner sep=0pt,minimum size=2mm,label=left:{$k-3$}] (u2) at (0,-1)   {};
\node[fill=none,draw=none,circle,inner sep=0pt,minimum size=0mm,label=left:{$\vdots$}] (u3) at (0,-2)   {};
\node[fill=white,draw,circle,inner sep=0pt,minimum size=2mm,label=left:{$2$}] (u4) at (0,-3)   {};
\node[fill=white,draw,circle,inner sep=0pt,minimum size=2mm,label=left:{$k$}] (u5) at (0,-4)   {};
\node[fill=white,draw,circle,inner sep=0pt,minimum size=2mm,label=left:{$k-2$}] (u6) at (0,-5)   {};
\node[fill=none,draw=none,circle,inner sep=0pt,minimum size=0mm,label=left:{$\vdots$}] (u7) at (0,-6)   {};
\node[fill=white,draw,circle,inner sep=0pt,minimum size=2mm,label=left:{$3$}] (u8) at (0,-7)   {};
\node[fill=white,draw,circle,inner sep=0pt,minimum size=2mm,label=left:{$1$}] (u9) at (0,-8)   {};
 
\draw(u1) -- (u2);
\draw[dashed](u2) -- (u3);
\draw[dashed](u3) -- (u4);
\draw(u4) -- (u5);
\draw(u5) -- (u6);
\draw[dashed](u6) -- (u7);
\draw[dashed](u7) -- (u8);
\draw(u8) -- (u9);
 \end{scope}
 
 \begin{scope}[scale=0.5,xshift=6cm]
\node[fill=white,draw,circle,inner sep=0pt,minimum size=2mm,label=left:{$k$}] (u1) at (0,0)   {};
\node[fill=white,draw,circle,inner sep=0pt,minimum size=2mm,label=left:{$k-2$}] (u2) at (0,-1)   {};
\node[fill=none,draw=none,circle,inner sep=0pt,minimum size=0mm,label=left:{$\vdots$}] (u3) at (0,-2)   {};
\node[fill=white,draw,circle,inner sep=0pt,minimum size=2mm,label=left:{$3$}] (u4) at (0,-3)   {};
\node[fill=white,draw,circle,inner sep=0pt,minimum size=2mm,label=left:{$1$}] (u5) at (0,-4)   {};
\node[fill=white,draw,circle,inner sep=0pt,minimum size=2mm,label=left:{$k-1$}] (u6) at (0,-5)   {};
\node[fill=none,draw=none,circle,inner sep=0pt,minimum size=0mm,label=left:{$\vdots$}] (u7) at (0,-6)   {};
\node[fill=white,draw,circle,inner sep=0pt,minimum size=2mm,label=left:{$4$}] (u8) at (0,-7)   {};
\node[fill=white,draw,circle,inner sep=0pt,minimum size=2mm,label=left:{$2$}] (u9) at (0,-8)   {};
 
\draw(u1) -- (u2);
\draw[dashed](u2) -- (u3);
\draw[dashed](u3) -- (u4);
\draw(u4) -- (u5);
\draw(u5) -- (u6);
\draw[dashed](u6) -- (u7);
\draw[dashed](u7) -- (u8);
\draw(u8) -- (u9);
 \end{scope}
 
 \begin{scope}[scale=0.5,xshift=12cm]
\node[fill=white,draw,circle,inner sep=0pt,minimum size=2mm,label=left:{$1$}] (u1) at (0,0)   {};
\node[fill=white,draw,circle,inner sep=0pt,minimum size=2mm,label=left:{$2$}] (u2) at (0,-1)   {};
\node[fill=none,draw=none,circle,inner sep=0pt,minimum size=0mm,label=left:{$\vdots$}] (u3) at (0,-2)   {};
\node[fill=white,draw,circle,inner sep=0pt,minimum size=2mm,label=left:{$\frac{k-2}{2}$}] (u4) at (0,-3)   {};
\node[fill=white,draw,circle,inner sep=0pt,minimum size=2mm,label=left:{$\frac{k}{2}$}] (u5) at (0,-4)   {};
\node[fill=white,draw,circle,inner sep=0pt,minimum size=2mm,label=left:{$\frac{k+2}{2}$}] (u6) at (0,-5)   {};
\node[fill=none,draw=none,circle,inner sep=0pt,minimum size=0mm,label=left:{$\vdots$}] (u7) at (0,-6)   {};
\node[fill=white,draw,circle,inner sep=0pt,minimum size=2mm,label=left:{$k-2$}] (u8) at (0,-7)   {};
\node[fill=white,draw,circle,inner sep=0pt,minimum size=2mm,label=left:{$k-1$}] (u9) at (0,-8)   {};
\node[fill=white,draw,circle,inner sep=0pt,minimum size=2mm,label=left:{$k$}] (u10) at (0,-9)   {};
 
\draw(u1) -- (u2);
\draw[dashed](u2) -- (u3);
\draw[dashed](u3) -- (u4);
\draw(u4) -- (u5);
\draw(u5) -- (u6);
\draw[dashed](u6) -- (u7);
\draw[dashed](u7) -- (u8);
\draw(u8) -- (u9);
\draw(u9) -- (u10);
\end{scope}

\begin{scope}[scale=0.5,xshift=15cm,yshift=-1cm]
\node[fill=white,draw,circle,inner sep=0pt,minimum size=2mm,label=left:{$k$}] (u0) at (0,1)   {};
\node[fill=white,draw,circle,inner sep=0pt,minimum size=2mm,label=left:{$k-2$}] (u1) at (0,0)   {};
\node[fill=white,draw,circle,inner sep=0pt,minimum size=2mm,label=left:{$k-4$}] (u2) at (0,-1)   {};
\node[fill=none,draw=none,circle,inner sep=0pt,minimum size=0mm,label=left:{$\vdots$}] (u3) at (0,-2)   {};
\node[fill=white,draw,circle,inner sep=0pt,minimum size=2mm,label=left:{$2$}] (u4) at (0,-3)   {};
\node[fill=white,draw,circle,inner sep=0pt,minimum size=2mm,label=left:{$k-1$}] (u5) at (0,-4)   {};
\node[fill=white,draw,circle,inner sep=0pt,minimum size=2mm,label=left:{$k-3$}] (u6) at (0,-5)   {};
\node[fill=none,draw=none,circle,inner sep=0pt,minimum size=0mm,label=left:{$\vdots$}] (u7) at (0,-6)   {};
\node[fill=white,draw,circle,inner sep=0pt,minimum size=2mm,label=left:{$3$}] (u8) at (0,-7)   {};
\node[fill=white,draw,circle,inner sep=0pt,minimum size=2mm,label=left:{$1$}] (u9) at (0,-8)   {};

\draw(u1) -- (u2);
\draw[dashed](u2) -- (u3);
\draw[dashed](u3) -- (u4);
\draw(u4) -- (u5);
\draw(u5) -- (u6);
\draw[dashed](u6) -- (u7);
\draw[dashed](u7) -- (u8);
\draw(u8) -- (u9);
\draw(u0) -- (u1);
 \end{scope}

 \begin{scope}[scale=0.5,xshift=18cm,yshift=-1cm]
 \node[fill=white,draw,circle,inner sep=0pt,minimum size=2mm,label=left:{$k$}] (u0) at (0,1)   {};
\node[fill=white,draw,circle,inner sep=0pt,minimum size=2mm,label=left:{$k-1$}] (u1) at (0,0)   {};
\node[fill=white,draw,circle,inner sep=0pt,minimum size=2mm,label=left:{$k-3$}] (u2) at (0,-1)   {};
\node[fill=none,draw=none,circle,inner sep=0pt,minimum size=0mm,label=left:{$\vdots$}] (u3) at (0,-2)   {};
\node[fill=white,draw,circle,inner sep=0pt,minimum size=2mm,label=left:{$3$}] (u4) at (0,-3)   {};
\node[fill=white,draw,circle,inner sep=0pt,minimum size=2mm,label=left:{$1$}] (u5) at (0,-4)   {};
\node[fill=white,draw,circle,inner sep=0pt,minimum size=2mm,label=left:{$k-2$}] (u6) at (0,-5)   {};
\node[fill=none,draw=none,circle,inner sep=0pt,minimum size=0mm,label=left:{$\vdots$}] (u7) at (0,-6)   {};
\node[fill=white,draw,circle,inner sep=0pt,minimum size=2mm,label=left:{$4$}] (u8) at (0,-7)   {};
\node[fill=white,draw,circle,inner sep=0pt,minimum size=2mm,label=left:{$2$}] (u9) at (0,-8)   {};

\draw(u1) -- (u2);
\draw[dashed](u2) -- (u3);
\draw[dashed](u3) -- (u4);
\draw(u4) -- (u5);
\draw(u5) -- (u6);
\draw[dashed](u6) -- (u7);
\draw[dashed](u7) -- (u8);
\draw(u8) -- (u9);
\draw(u0) -- (u1);
 \end{scope}

 \node[fill=none,draw=none,circle,inner sep=0pt,minimum size=2mm,label=center:{$k$ odd}] (u11) at (1.5,0.5)   {};
 \node[fill=none,draw=none,circle,inner sep=0pt,minimum size=2mm,label=center:{$k$ even}] (u12) at (7.5,0.5)   {};
 \end{tikzpicture}
}

\usepackage{authblk}

\title{Minimizing Maximum Dissatisfaction in the Allocation of Indivisible Items under a Common Preference Graph}

\author[1,2]{Nina Chiarelli}
\author[3]{Cl\'ement Dallard}
\author[4]{Andreas Darmann}
\author[4]{Stefan Lendl}
\author[1,2]{Martin Milani\v c}
\author[1]{Peter Mur\v si\v c}
\author[4]{Ulrich Pferschy}

\affil[1]{FAMNIT, University of Primorska, Koper, Slovenia}
\affil[2]{IAM, University of Primorska, Koper, Slovenia}
\affil[3]{Department of Informatics, University of Fribourg, Fribourg, Switzerland}
\affil[4]{Department of Operations and Information Systems, University of Graz, Austria}

\date{}

\begin{document}

\maketitle

\begin{abstract}
    We consider the task of allocating indivisible items to agents, when the agents' preferences over the items are identical. 
    The preferences are captured by means of a directed acyclic graph, with vertices representing items and an edge $(a,b)$, meaning that each of the agents prefers item $a$ over item~$b$. 
    The dissatisfaction of an agent is measured by the number of items that the agent does not receive and for which it also does not receive any more preferred item.
    The aim is to allocate the items to the agents in a fair way, i.e., to minimize the maximum dissatisfaction among the agents. 
    We study the status of computational complexity of that problem and establish the following dichotomy:
    the problem is \NP-hard for the case of at least three agents, even on fairly restricted graphs, but polynomially solvable for two agents.
    We also provide several polynomial-time results with respect to different underlying graph structures, such as graphs of width at most two and tree-like structures such as stars and matchings. These findings are complemented with fixed parameter tractability results related to path modules and independent set modules. Techniques employed in the paper include bottleneck assignment problem,  greedy algorithm, dynamic programming, maximum network flow, and integer linear programming.
    

\end{abstract}
\maketitle

\section{Introduction}

Partitioning resources or allocating indivisible items to a set of agents is a widely studied research topic  in Operations Research and in particular in discrete optimization (see, e.g., \cite{aziz20,bouveret-survey,thomson,kilgour2018,cornilly2022}).
In many scenarios agents can voice their preferences over the set of available items.
Frequently, this is done by assigning a profit value to every item.
Under the commonly applied additivity assumption (e.g., 
\cite{lipton04,proc14}), the total utility or gain that an agent obtains from a certain subset of items is simply given by the sum of the corresponding profit values.

However, it is often pointed out that individuals would be overwhelmed with choosing meaningful numerical values for a larger set of items (cf.\ \cite{ALOYSIUS2006273,BANAECOSTA1994489}).
Moreover, cardinal evaluations entail certain drawbacks compared to ordinal evaluations as widely discussed in theoretical economics, see, e.g.,  \cite{hamm91}.
Instead, it would be easier to ask agents for pairwise comparisons and the resulting orders as in \cite{aziz}.
For highly heterogenous items a total order may also be hard to find since certain items could well be incomparable to each other.
Thus, one often settles for partial orders to represent preferences of agents. 
It is also fairly natural to consider agents discussing their opinions of items, in particular if these are not so well known, and reaching a consensus valuation expressed by means of a common preference graph. The setting of agents expressing common, i.e., identical, preferences has been considered in the literature on  fair division of indivisible items in several works, including the ones by Bouveret and Lang~\cite{bouveretlang2008}, Brams and Fishburn~\cite{brams2000}, and Freeman et al.~\cite{freeman}.

\medskip
In this contribution we consider a set of $k$ agents denoted by $K=\{1,\ldots,k\}$ and a set $V$ of $n$ items.
The aforementioned preference structure of the considered problem is described by a single {\em preference graph}, i.e., a directed acyclic graph~$G$.
This is the crucial difference of the current contribution from its predecessor paper \cite{general}, where each agent had its own preference graph.
Every item is identified with a vertex of~$G$.
An arc $(a,b)$ in $G$ means item $a$ is preferred over item~$b$. 
Assuming transitivity of the preferences, arcs $(a,b)$ and $(b,c)$ imply that the agents also prefer item $a$ over item $c$, regardless of whether the arc $(a,c)$ is contained in the graph or not. 
Obviously, $G$ induces a partial order over the set of items.

\medskip
The allocation of items is performed by a central decision maker who assigns every item to at most one agent. 
Now a crucial question arises:
What is the happiness or satisfaction of an agent receiving a certain subset of items under the given partial preference order?

We will consider a new measure for expressing the \emph{dissatisfaction} of an agent with its allocated subset of items
which was recently introduced in our earlier work~\cite{general}.
The underlying rationale states that an item received makes all lesser ranked items irrelevant. 
However, an item not assigned to a certain agent causes dissatisfaction if this agent does not receive any other more preferred item.
The overall dissatisfaction of an agent is therefore given by the total number of items causing dissatisfaction.

The global goal, then, would be a {\em fair} allocation. 
In that respect, in order to evaluate the quality of an allocation, we take into account a very natural and  basic fairness criterion: 
we aim at making the worst-off agent as good as possible, i.e., our problem
{\MINMAX} (see~\Cref{sec:prelim} for a precise definition) deals with minimizing the maximum dissatisfaction value of an agent for a common preference graph.
Thus, in principle, we are concerned with maximizing egalitarian social welfare on the basis of  the agents' preferences as considered, e.g., by Bez\'{a}kov\'{a} and Dani~\cite{bezakova2005}, Baumeister et al.~\cite{baumeister2017}, and Roos and Rothe~\cite{roosarothe}. 
Egalitarian social welfare as a fairness measure has a long-standing tradition (see Rawls~\cite{Rawls1971}), and has been applied in various settings (see, e.g., Golovin~\cite{Golovin05}, Kawase and Sumita~\cite{kawase22}, Darmann et al.~\cite{darmann2009} and Nguyen et al.~\cite{nguyen}).

From a graph theoretic perspective, an allocation for an agent $i \in K$ is evaluated by the number of vertices in $G=(V,A)$ which are {\em dominated} by the allocated items, i.e., items which can be reached from an allocated vertex by a directed path in~$G$. 
The total number of these dominated items, together with the allocated items, can be seen as the {\em satisfaction} level of agent~$i$.
Obviously, for each agent satisfaction and dissatisfaction add up to the number of vertices in~$G$. 

\subsection{Our contribution}

In Section~\ref{sec:general}, we prove \NP-hardness of {\MINMAX} for any constant number of agents $k \geq 3$. 
This hardness result even holds for the relatively basic class of preference graphs, namely graphs with no directed path of length two and no vertex of in-degree larger than two. 
Therewith we provide a dichotomy for the computational complexity of  {\MINMAX} with respect to the number of agents, as the problem is solvable in polynomial time in the case of two agents. 
Moreover, we show that in contrast to the above-mentioned hardness result for graphs of ``height'' two, {\MINMAX} can be solved in polynomial time in graphs of \textit{width} at most two, for any number of agents involved.  

Turning to tree-like structures (Section~\ref{sec:minmax}), we then show that {\MINMAX} is solvable in polynomial time not only for directed matchings but also, more generally, when the preference graph is a collection of out-stars. In addition, for a constant number of agents, the polynomial-time result for out-stars can be generalized to out-forests. 

In Section~\ref{sec:modular}, we then provide a fixed parameter tractability (FPT) result for a bounded number of agents when the preference graph can be decomposed into path modules and independent set modules.  
For the case of a decomposing into independent set modules only, fixed parameter tractability holds for an arbitrary number of agents.

The algorithms employed in this paper encompass a broad spectrum of classical ingredients, such as bottleneck assignment problem,  greedy algorithm, dynamic programming, maximum network flow, and integer linear programming.

\medskip
It should be noted that a generalization from the common preference graph to the setting where each agent has its own preference graph makes the problem much harder in terms of computational complexity. As shown in~\cite{general}, for individual preference graphs the problem of minimizing the maximum dissatisfaction can be solved in polynomial time if each preference graph is a path, as well as for any constant number of agents if the underlying undirected graph of the union of the preference graphs has bounded treewidth (see~\cite[Remark 18]{general} for details).
Note that this result implies that for any constant number of agents, {\MINMAX} is solvable in polynomial time if the underlying graph of the preference graph has bounded treewidth.

On the other hand, it was also shown in~\cite{general} that for individual preference graphs, the problem of minimizing the maximum dissatisfaction remains \NP-complete in each of the following cases:
\begin{enumerate}[(1)]
    \item for two agents, even if the two sets of items are the same,\Omit{{\color{red}~and each preference graph is of ``height'' two (that is, contains no two-edge directed path)}}
    \item if each preference graph is an out-star,
    \item if each preference graph is a directed matching, and
    \item if each preference graph consists of at most two paths containing at most five items in total.
\end{enumerate}
As a consequence of the results of the present paper, all of these intractable cases turn out to be solvable in polynomial time in the case of a common preference graph.

Finally, we point out that the graph classes for which our polynomial-time results hold cover rather natural and basic ways of expressing preferences. 
For practical purposes, these might be easier to state than, e.g., a strict ranking over all the available items. 
For instance, a directed matching represents the situation in which the agents' preferences are expressed over disjoint pairs of items only; i.e., the agents prefer, for each $i$, item $a_i$ over $b_i$, with no additional preferences stated between the items. An out-star covers the scenario in which each agent prefers a dedicated item to several other items, with no further preferences between the remaining items. And an out-forest, as it is a collection of out-stars, (potentially) contains several such scenarios. 

Some of the results of this paper were included or announced in the proceedings paper~\cite{adt2021}.

\section{Preliminaries}\label{sec:prelim}

For the sake of consistency, definitions and notation mostly follow those from~\cite{general}.
We consider graphs that are finite and do not contain loops or multiple edges.
For brevity, we often say \textit{graph} when referring to a directed graph. 
When undirected graphs are used, this will be explicitly stated.

Consider a directed graph $G=(V,A)$ with~$n:=|V|$.
For $a=(u,v)\in A$, vertex $u$ is called the \emph{tail} of $a$ and vertex $v$ is the \emph{head} of~$a$. 
The \textit{in-degree} of a vertex $u$ is the number of arcs in $A$ for which $u$ is the head and the \textit{out-degree} of $u$ is the number of arcs in $A$ for which $u$ is the tail.
The \textit{degree} of a vertex $u$ is the number of arcs in $A$ for which $u$ is either head or tail.
A vertex with in-degree $0$ is called \textit{source}. A vertex with out-degree $0$ is called \textit{sink}. 

A sequence $p=(v_0,v_1,v_2,\ldots,v_\ell)$ with $\ell \ge 0$
and $(v_i,v_{i+1})\in A$ for each $i\in \{0,\ldots,\ell-1\}$ is called a \textit{walk} of length $\ell$ from $v_0$ to $v_\ell$; it is a \textit{path} (of length $\ell$ from $v_0$ to $v_\ell$) if all its vertices are pairwise distinct.
A walk from  $v_0$ to $v_\ell$ is \textit{closed} if $v_0 = v_\ell$.
A \textit{cycle} is a closed walk of positive length in which all vertices are pairwise distinct, except that $v_0 = v_\ell$.

A \textit{directed acyclic graph} is a directed graph with no cycle. Observe that in a directed acyclic graph there is always at least one source and at least one sink. 
An  \textit{out-tree} is a directed acyclic graph $G=(V,A)$ with a specified vertex $r$ (called \textit{root}) such that for each vertex $v\in V\setminus\{r\}$ there is exactly one path from $r$ to~$v$.  
An \textit{out-star} is an out-tree in which each such path is of length at most~$1$. 
An \textit{out-forest} is a disjoint union of out-trees.

\medskip
Let $G = (V,A)$ be a directed acyclic graph.
An \emph{antichain} in $G$ is a set of vertices that are pairwise unreachable from each other.
The \emph{width} of $G$, denoted by $w(G)$, is the maximum cardinality of an antichain in~$G$.
A \emph{chain} in $G$ is a set $C$ of vertices such that for every two vertices $x,y\in C$, the graph $G$ contains either an $x,y$-path or a $y,x$-path.
A \emph{chain partition} of $G$ is a family $\mathcal{C}$ of vertex-disjoint chains in $G$ such that every vertex of $G$ belongs to precisely one chain in $\mathcal{C}$.

\medskip
A directed acyclic graph $G=(V,A)$ induces a binary relation $\succ$ over $V$, by setting $u\succ v$ if and only if $u\neq v$ and there is a directed path from $u$ to $v$ in $G$; in such a case, we say that the agents prefer $u$ to~$v$.  
In particular, the transitive closure of a directed acyclic graph induces a strict partial order over~$V$.

 Let ${\it \pred}(v)$ denote the set of predecessors of $v$ in graph $G$, i.e., the set of all vertices $u\neq v$ such that there is a path from $u$ to $v$ in~$G$. 
 In addition, let ${\it \succc}(v) \subseteq V$ denote the set of successors of $v$ in graph $G$, i.e., the set of all vertices $u\neq v$ such that there is a path from $v$ to $u$ in~$G$. 
 We denote by $\pred[v]$ the set 
$\pred(v) \cup \{v\}$; similarly we denote by $\text{succ}[v]$ the set $\text{succ}(v) \cup \{v\}$.
For $u,v \in V$ we say that item $u$ is \textit{dominated} by  item $v$ if $v\in {\pred}[u]$. 
  
For a directed acyclic graph $G = (V,A)$ and a vertex $v\in V$, we denote by $N^-(v)$ the set of all in-neighbors of $v$, formally $N^-(v)=\{u\in V\colon (u,v)\in A\}$, and similarly by $N^+(v)$ the set of all out-neighbors of $v$, that is $N^+(v)=\{u\in V\colon (v,u)\in A\}$.

\medskip
As stated before, a set of $k$ agents is denoted by $K=\{1,\ldots,k\}$.
An \textit{allocation} $\pi$ is a function $K \rightarrow  2^V$ that assigns to the agents pairwise disjoint sets of items, i.e., for $i,j \in K$, $i\not=j$, we have $\pi(i)\cap\pi(j)=\emptyset$.

To measure the attractiveness of an allocation we will count the number of items that an agent {\bf does not receive} and for which it receives no other more preferred item.
Formally, for an allocation $\pi$ the \textit{dissatisfaction} $\diss_{\pi}(i)$ of agent $i$ is defined as the number of items in $G$ not dominated by any item in $\pi(i)$.  
If $\pi$ allocates a vertex $v$ to agent $i$, it will not change the dissatisfaction $\diss_{\pi}(i)$ if any vertices in $\text{succ}(v)$ are allocated to $i$ in addition to~$v$.
Thus, we say that if $v \in \pi(i)$, then agent $i$ \textit{dominates} all vertices in $\text{succ}[v]$. It will also be convenient to define the \textit{satisfaction} $s_\pi (i)$ of agent $i$ with respect to allocation $\pi$ 
as the number of items in $V$ that are dominated by~$i$.

Continuing the work of \cite{general} we consider the minimization of the  maximum  dissatisfaction among the agents. As a decision problem, this task is formulated as follows.

\medskip
\probdef[Question]{A set $K$ of agents, a set $V$ of items, a directed  acyclic graph $G =(V,A)$, and an integer~$d$.}{Is there an allocation $\pi$ of items to agents such that the dissatisfaction $\diss_{\pi}(i)$ is at most $d$ for each agent $i\in K$?}{\textsc{Min-Max Dissatisfaction with a Common Preference Graph}}{}{}

We abbreviate the above decision problem as \MINMAX. 
Observe that minimizing the maximum dissatisfaction of an agent is equivalent to maximizing the minimum satisfaction. 
We will often use this fact in some of our proofs, without explicitly referring to it.
As pointed out in \cite{general} this equivalence does not hold anymore if each agent has its own preference graph.

By \textit{optimal allocation} and \textit{optimal solution value} we refer to the optimization version of \MINMAX, i.e., the task of actually determining an allocation that minimizes the maximum dissatisfaction among agents and the corresponding maximum dissatisfaction value $\delta$. 
(For the equivalent problem of maximizing the minimum satisfaction among agents the minimum satisfaction is $n-\delta$.)
For the polynomial-time results in this paper we provide constructive proofs, i.e., our results also yield that a respective optimal allocation can be found in polynomial time. 
It is easy to see that if there are fewer items than agents, then there must be at least one agent receiving no item at all and thus reaching the worst possible dissatisfaction~$n=|V|$.

\begin{observation}\label{th:kgreatern}
If $k > n$, then {\MINMAX} has a canonical solution independent of~$G$.
\end{observation}

Based on Observation~\ref{th:kgreatern} we will asume $k \leq n$ for the remainder of the paper.

\medskip
A classical theorem of Dilworth states that the width of $G$ equals the minimum number of chains in a chain partition of $G$~\cite{Dilworth}.
Moreover, by applying the approach of Fulkerson~\cite{MR0078334}, a minimum chain partition of $G$ can be computed efficiently by solving a maximum matching problem in a derived undirected bipartite graph having $2n$ vertices (see~\cite{MR545530}). 
This can be done in time $\mathcal{O}(n^{5/2})$ using the algorithm of Hopcroft and Karp~\cite{MR0337699}. 

Many of our proofs will make use of the following general lemma on the solutions to {\MINMAX}.

\begin{lemma}\label{structure-of-optimal-solutions}
When solving {\MINMAX} for a preference graph $G$ with $n$ vertices and any set $K$ of $k$ agents such that $k\le n$, we may without loss of generality restrict our attention to allocations $\pi$ such that for each agent $i\in K$, the set of items allocated to agent $i$ forms a nonempty antichain in~$G$.
\end{lemma}

\begin{proof}
If there exists an agent $i\in K$ such that $\pi(i)$ is not an antichain in $G$, we can simply remove any item allocated to agent $i$ that is dominated by some other item in $\pi(i)$, without affecting the dissatisfaction of agent~$i$.
Repeating this procedure if necessary, we obtain an allocation $\pi$ such that for each agent $i\in K$, the set of items allocated to agent $i$ forms an antichain in~$G$.

Let us now argue why we may assume that all these antichains are nonempty. 
For any allocation $\pi$ in which some agent $i\in K$ receives no items, the dissatisfaction of this agent is equal to $n$, and hence the maximum dissatisfaction equals~$n$.
On the other hand, since $k\le n$, there exists an allocation in which every agent receives at least one item.
Any such allocation results in the maximum dissatisfaction at most $n-1$ and is thus strictly preferred to any allocation assigning no items to some agent.
\end{proof}

The classical \textsc{Linear Sum Assignment} problem asks for a perfect matching with minimum total weight in a bipartite graph, and the \textsc{Linear Bottleneck Assignment} problem seeks a perfect matching in a weighted bipartite graph such that the largest weight of a matching edge is as small as possible.
Both problems can be solved in polynomial time (see \cite[Sec.~6.2]{assign2012}).
For our purposes we will require an assignment restricted to $k$ matching edges. 
This variant was considered in \cite{Amico97} for the sum case (i.e., minimizing the total weight).
Moreover, it is easy to see that the $k$-cardinality restriction can be obtained by adding $n-k$ dummy vertices to each set of the bipartition. 
Without going into details of this exercise we state the result for the min-max case.

\begin{lemma}[folklore]\label{matching-k-min-sum}
Given a bipartite graph $G = (V=A\sqcup B ,E)$ having a perfect matching, an edge weight function $w:E\to \mathbb{R}_+$, and an integer $k\le |V|/2$,  the following problem can be solved in polynomial time:
Compute a matching in $G$ with cardinality $k$ such that the largest weight of a matching edge is as small as possible.
\end{lemma}

\section{General results for {\MINMAX}}\label{sec:general}

Our first result states that for two agents, {\MINMAX} can be solved in polynomial time.

\begin{theorem}\label{th:2agents}
For $k=2$ and any preference graph $G = (V,A)$, {\MINMAX} can be solved in $\mathcal{O}(|V|+|A|)$ time.
\end{theorem}

\begin{proof}
\begin{sloppypar}
Let $S$ be the set of sources of $G$ (these correspond to items that are not dominated by any other item).
Let $G'$ be the graph $G-S$ and let $S'$ be the set of sources of~$G'$.
Furthermore, let $S_1$ be any subset of $S$ with cardinality $\left\lfloor \frac{|S|}{2} \right\rfloor$ and $S_2=S\setminus S_1$.
We denote by $S_1'$ the vertices in $S'$ that are not dominated by any vertex in $S_1$,
and, similarly, by $S_2'$ the vertices in $S'$ that are not dominated by any vertex in~$S_2$.
More formally, $S_1' = S'\setminus \{{\it \succc}(v)\colon v\in S_1\}$ and
$S_2' = S'\setminus \{{\it \succc}(v)\colon v\in S_2\}$.
The four sets $S_1$, $S_2$, $S_1'$, and $S_2'$ can be computed in time linear in the size of the graph by a simple graph traversal.
Moreover, we claim that they are pairwise disjoint.
The disjointness of any pair follows immediately from the definitions, except for $S_1'$ and~$S_2'$.
Suppose for a contradiction that there exists a vertex $u\in S_1'\cap S_2'$.
Then $u$ belongs to $S'$ and, hence, is a source in~$G-S$.
Since $u$ is not a source in $G$, it must have a predecessor $v$ in~$G$.
Furthermore, since $u$ is a source in $G-S$, vertex $v$ must belong to $S$; in particular, we must have
$v\in S_i$ for some $i\in \{1,2\}$. However, this implies that $u\not\in S_i'$, a contradiction.
\end{sloppypar}

Since the sets $S_1\cup S_1'$ and $S_2\cup S_2'$ are disjoint, setting $\pi(1) = S_1\cup S'_1$
and $\pi(2) = S_2\cup S'_2$ defines a valid allocation of items to the two agents.
The undominated items for agent $1$ (agent $2$) are exactly the items in $S_2$ (in $S_1$).
Clearly, $\pi$ can be computed in linear time and the corresponding dissatisfaction of the two agents is
$\delta_\pi(1) = |S_2| = \left\lceil \frac{|S|}{2} \right\rceil$ and
$\delta_\pi(2)= |S_1| = \left\lfloor \frac{|S|}{2} \right\rfloor$.
To complete the proof, we show that $\pi$ in fact optimally solves
{\MINMAX} problems for the given input instance~$G^*$.

Note that under allocation $\pi$, the dissatisfaction of each agent is at most $\left\lceil \frac{|S|}{2} \right\rceil$.
Since for any allocation $\pi^*$ we have $\delta_{\pi^*}(1)+ \delta_{\pi^*}(2)\geq |S|$ and consequently $\max\{\delta_{\pi^*}(1), \delta_{\pi^*}(2)\} \geq \left\lceil \frac{|S|}{2} \right\rceil$, the allocation $\pi$ is indeed optimal for {\MINMAX}.
\end{proof}

The result of Theorem~\ref{th:2agents} is best possible regarding the number of agents, unless \PP{} = \NP.
Indeed, we show next that {\MINMAX} is \NP-complete for $k$ agents for any $k\ge 3$, even if the graph has no directed path of length two. 

\begin{theorem}\label{th:3agents}
For each $k\ge 3$, {\MINMAX} is \NP-complete.
This result holds even if the preference graph has no directed path of length two and no vertex of in-degree larger than 2.
\end{theorem}

\begin{proof}\begin{sloppypar}
Membership in $\NP$ is clear.
We provide a reduction from \textsc{$k$-Colorability}, which is the following decision problem: ``Given an undirected graph $H$, is $H$ $k$-colorable?''
This problem is well known to be \NP-complete for every $k\geq 3$ (Garey et al.~\cite{gareyjohnstock76}).  We proceed as follows. Given an instance of the above problem, we create a directed graph $G$ from $H$ by doing the following for each edge $e = \{u,v\}$ of $H$: delete the edge $e$ and add a new vertex $w_e$ and arcs~$(u,w_e),(v,w_e)$.
In words, subdivide every edge of $H$ and orient the edges of the so-obtained graph towards the new vertices.
Let the graph $kG$ be made up of $k$ disjoint copies of~$G$. Below we show that $H$ is $k$-colorable if and only if the graph $kG$ admits an allocation for $k$ agents such that each agent has dissatisfaction of at most $\totdiss :=k|V(G)|-\sum_{v \in V(G)}|\succc[v]|$. 
\end{sloppypar}
Assume that such an allocation $\pi$ for $kG$ exists. Observe that for each copy of $G$, the total dissatisfaction (i.e., sum over all agents' dissatisfaction) of any allocation is \emph{at least} $\totdiss$:
without being allocated, each vertex $v \in V(G)$ creates a unit of dissatisfaction for each agent (hence the term $k|V(G)|$ in $\totdiss$); if vertex $v$ is allocated to an agent, the total dissatisfaction value is reduced by at most $|\succc[v]|$.
A reduction by \emph{exactly} $|\succc[v]|$ is achieved if none of the successors of $v$ has a predecessor allocated to the same agent. 
Summing up over all $v\in V(G)$ yields the lower bound of $\totdiss$.

By our assumption, allocation $\pi$ yields a total dissatisfaction over all $k$ agents of at most $k\cdot \totdiss$.
It follows from above that $kG$ contributes at least $k\cdot \totdiss$ to the total dissatisfaction.
Therefore, equality must hold.
As indicated above, for such an allocation $\pi$ with total dissatisfaction exactly $k\cdot \totdiss$ it must hold for every vertex $w_e$ arising from an edge $e = \{u,v\}$  that $w_e$ and its predecessors $u$ and $v$ are assigned to three different agents.
Thus, identifying colors with agents, a $k$-coloring for $H$ results from $\pi$  restricted to an arbitrary copy of $G$ in the graph~$kG$. 

Conversely, a $k$-coloring for $H$ can be extended to an allocation for $kG$ as follows. 
Pick a copy of $G$ in~$kG$. In that copy, identify colors with agents and allocate each vertex of color $h$ to agent~$h$. In addition, for each edge $e\in E(H)$, allocate the item $w_{e}$ to any of the $k-2$ agents that did not receive an endpoint of~$e$. 
Such an allocation for $G$ has a total dissatisfaction of $k|V(G)|-\sum_{v \in V(G)}|\succc[v]|= \totdiss$, since no vertex assigned to an agent $h$ is dominated by another vertex allocated to~$h$. 
The extended allocation for $kG$ can be obtained by ``cycling'' the agents in each copy:
If a vertex belongs to agent $i$ in copy $j$ of $G$, then this vertex should belong to agent $i+\ell \pmod k$ in copy $j+\ell \pmod k$ of~$G$. 
The total dissatisfaction in $kG$ is hence $k\cdot \totdiss$. Note that every agent gets the same dissatisfaction due to the cycling, hence the individual dissatisfaction equals to $\totdiss$ for each agent.
This completes the reduction.
\end{proof}

In contrast to the intractability results for {\MINMAX} on graphs with ``height'' two as shown by \Cref{th:3agents}, we now provide an efficient algorithm for the problem on graphs of \emph{width} at most two (see \Cref{fig:width-2-graph} for an example).

\begin{figure}
    \centering
\begin{tikzpicture}[every node/.style={draw,circle,inner sep=0.8pt,minimum size=2mm}]
\foreach \x in {0,...,7}{
    \node (A\x) at (1,-\x) {};
}
\foreach \x in {0,...,5}{
    \node (B\x) at (3,-\x) {};
}

\draw[->] (A0) -> (A1);
\draw[->] (A1) -> (A2);
\draw[->] (A2) -> (A3);
\draw[->] (A3) -> (A4);
\draw[->] (A4) -> (A5);
\draw[->] (A5) -> (A6);
\draw[->] (A6) -> (A7);

\draw[->] (B0) -> (B1);
\draw[->] (B1) -> (B2);
\draw[->] (B2) -> (B3);
\draw[->] (B3) -> (B4);
\draw[->] (B4) -> (B5);

\draw[->] (A0) -> (B1);
\draw[->] (A1) -> (B2);
\draw[->] (A1) -> (B3);

\draw[->] (B1) -> (A3);

\draw[->] (B4) -> (A6);

\draw[->] (B5) -> (A7);

\end{tikzpicture}
    \caption{Illustration of a graph of width~$2$.}
    \label{fig:width-2-graph}
\end{figure}
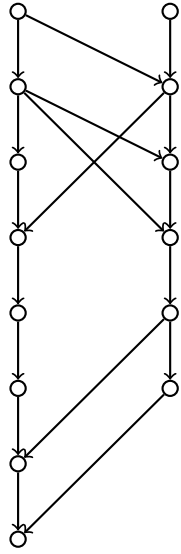

\begin{theorem}\label{th:width-two-MINMAX-MINSUM}
{\MINMAX} is solvable in polynomial time for any number $k$ of agents if the preference graph $G$ has width at most~$2$.
\end{theorem}

\begin{proof}
Let $G = (V,A)$ be a directed acyclic graph of width at most~$2$.
Let us associate to $G$ an undirected bipartite graph $\widehat G = (\widehat V,E)$ defined as follows:
\begin{enumerate}
\item The vertex set of $\widehat G$ is $\widehat V = V\cup V'$ where $V' = \{v'\colon v\in V\}$ is a set of new vertices.
\item There are two types of edges in $\widehat G$: edges of the form $\{x,y\}\subseteq V$, where $\{x,y\}$ is an antichain 
in $G$ with size two, and edges of the form $\{v, v'\}$, where $v\in V$.
\end{enumerate}
Note that $\widehat G$ is bipartite.
Indeed, if $G$ has width one, then each of $V$ and $V'$ is an independent set in $\widehat G$ and their union is $\widehat V$.
If $G$ has width two, then by Dilworth's theorem there exists a chain partition $\{P,Q\}$ of $G$ with size two, and the sets $P\cup \{v'\colon v\in Q\}$ and $Q\cup \{v'\colon v\in P\}$ are independent sets in $\widehat G$ with union $\widehat V$.
Furthermore, the edges of the form $\{v, v'\}$, where $v\in V$, form a perfect matching in $\widehat G$.

We now assign nonnegative weights to the edges of $\widehat G$ as follows.
For each edge $e = \{x,y\}\subseteq V$ corresponding to an antichain in $G$ with size two, we define the weight to $e$ to be the dissatisfaction of an agent $i\in K$ under any allocation $\pi$ such that $\pi(i) = \{x,y\}$, that is, the total number of vertices in $G$ that are not reachable by a path from $\{x,y\}$.
Similarly, for each edge of the form $\{v, v'\}$ where $v\in V$, we define the weight to $e$ to be the dissatisfaction of an agent $i\in K$ under any allocation $\pi$ such that $\pi(i) = \{v\}$, that is, the total number of vertices in $G$ that are not reachable by a path from~$v$.
Clearly, this weight function can be computed in polynomial time given~$G$.

When solving the {\MINMAX} problem for $G$ and $K$, \Cref{structure-of-optimal-solutions} implies that each agent $i\in K$ will be allocated either only one item or two items forming an antichain.
To each such allocation $\pi$ we can associate a matching $M_\pi$ in $\widehat{G}$ as follows.
For each agent $i\in K$ receiving only one item, say $\pi(i) = \{v\}$ for some $v\in V$, we include in $M_\pi$ the edge $\{v,v'\}$.
For each agent $i\in K$ receiving some antichain $\{x,y\}\subseteq V$ of size two, i.e., $\pi(i) = \{x,y\}$, we include in $M_\pi$ the edge $\{x,y\}$.
Since the sets of items assigned to different agents are pairwise disjoint, the obtained set $M_\pi$ is indeed a matching.
Furthermore, since for each of the $k$ agents, the set of items allocated to the agent forms a \emph{nonempty} antichain in $G$, we can associate to each agent $i\in K$ a unique edge $e_i\in M_{\pi}$; in particular, $M_{\pi}$ has cardinality~$k$.
By construction, the weight of the edge $e_i$ equals to the dissatisfaction of agent $i$ with respect to the allocation $\pi$.
Note that the above procedure can be reversed: to any matching $M = \{e_i\colon i\in K\}$ of cardinality $k$ in $\widehat G$, we can associate an allocation $\pi_M$ of items in $G$ to the $k$ agents such that the dissatisfaction of agent $i$ under $\pi_M$ equals the weight of the edge~$e_i$.

It follows that the {\MINMAX} problem for $G$ and $K$ can be reduced in polynomial time to the problem of computing a matching in $\widehat G$ with cardinality $k$ such that the maximum weight of a matching edge is as small as possible.
By~\Cref{matching-k-min-sum}, this problem can be solved in polynomial time.
\end{proof}

\section{Polynomially solvable cases on out-forests}\label{sec:minmax}

Starting with the simplest meaningful case of an out-forest, we consider the case where $G$ is a directed matching, i.e., a collection of disconnected edges.
It is easy to see that for this case {\MINMAX} can be solved by a greedy strategy,
even if the number of agents $k$ is unbounded.
It suffices to assign all heads of directed matching edges as equally as possible to the agents.
Then one can iteratively assign tails to agents with lowest score (avoiding assigning both ends of an edge to the same agent).
In this way, the difference between the dissatisfactions of any two agents can be at most one.
However, we can also extend this approach to the more general setting
of out-stars.

\begin{theorem}\label{th:minmax-out-stars}
For any number $k$ of agents, {\MINMAX} can be solved in polynomial time if the preference graph $G$ is a collection of out-stars.
\end{theorem}

\begin{proof}
To the input graph $G$ we associate a sorted sequence of integers $a_1 \geq a_2 \geq \dots \geq a_{\ell}\geq 1$ corresponding to the numbers of leaves of the \emph{nontrivial} out-stars (that is, out-stars with at least one edge), where $r_i$ is the root of a star with $a_i$ leaves, and an integer $a \geq 0$ denoting the number of trivial out-stars consisting of a single vertex.
Note that in the whole proof those are treated as leaves.
We show how to solve the optimization variant of the max-min satisfaction problem, which is equivalent to ${\MINMAX}$.
For the case $k=2$ the result follows immediately from Theorem~\ref{th:2agents}.
Let us consider $k\geq 3$.

We start by assigning the roots of the stars using the following greedy strategy. 
First, we process the nontrivial stars and assign their roots only.
For each $i\in \{1,\ldots,\ell\}$, denoting by $\pi'$ the partial allocation right before assigning the root $r_i$,
we select a least satisfied agent $j \in \arg\min \{ s_{\pi'}(j) \colon j \in K \}$ and extend $\pi'$ by assigning root $r_i$ to agent~$j$. 
Second, we distribute the leaves of the stars by applying the same greedy idea as for the roots, paired with an exchange argument. 
Consider a least satisfied agent~$j$. 
If there is still an unassigned leaf of a star that is either trivial or whose root has not been assigned to $j$, assign such a leaf to~$j$. 
Otherwise, if all trivial stars (i.e., singletons) are assigned and every unassigned leaf is adjacent to a root assigned to $j$, then, as we will show below, there exists an agent $h \neq j$ who has received a leaf $w$ of a star that is either trivial or whose root was not assigned to~$j$.
Given such an agent $h$, we pick an unassigned leaf $v$ (whose root was assigned to $j$) and perform an exchange by assigning $v$ to $h$ and, in turn, reallocating $w$ to~$j$.
In this way, the satisfaction of $h$ remains unchanged, while an additional leaf is assigned to~$j$.
Therefore, as long as there are unassigned leaves, it is always possible to assign a leaf to the least satisfied agent.

We now show that such an agent $h \neq j$ for the above exchange does indeed exist for $k \geq 3$.
Suppose, for the sake of contradiction, that such an agent does not exist. 
Then all leaves of stars whose roots are assigned to agents other than $j$ and all singletons must already be assigned to $j$ and $j$ is assigned at least one root. 
Hence, the least satisfied agent $j$ dominates all leaves of the nontrivial stars and all singletons and every agent in the generated allocation $\pi$ has satisfaction at least $1 +  a+\sum_{i=1}^{\ell} a_i$. 
As we will now show, this can only be the case if there are just two agents, contradicting our assumption of $k 
 \geq 3$. 
Observe that the total sum of agent satisfactions is at most $\ell + a + 2 \sum_{i=1}^{\ell} a_i$. 
Since every agent has satisfaction at least $1 +  a+\sum_{i=1}^{\ell} a_i$, we have that the total satisfaction is at least $k\left(1 +  a+\sum_{i=1}^{\ell} a_i\right)$, thus 
\begin{align*}
\sum_{j \in K}s_\pi(j)\geq k\left(1 +  a+\sum_{i=1}^{\ell} a_i\right) \geq 3 + 3a+ 3 \sum_{i=1}^{\ell} a_i =\\
=3+3a+ \sum_{i=1}^{\ell} a_i+ 2\sum_{i=1}^{\ell} a_i> a+ \ell+ 2\sum_{i=1}^{\ell} a_i \geq \sum_{j \in K}s_\pi(j),
\end{align*}
a contradiction.
Hence, agent $h$ must exist.

\medskip

At the end of this process, which clearly can be performed in polynomial time, all vertices of $G$ are assigned. In the remainder of the proof  we show  optimality of the allocation by distinguishing two cases.

\medskip

\noindent\textbf{Case 1.} \emph{Every agent receives at most one root by the greedy algorithm.}
Let $S^*$ be the (unknown) optimal value of the optimization problem of maximizing the minimum satisfaction among the agents. 
Let $R^*:= \{i \mid  a_i \geq S^*-1,  i \in \{1,\ldots,\ell\}\}$ denote the set of indices of roots giving, when assigned, a satisfaction at least~$S^*$.
W.l.o.g.\ we will assume that $|R^*|<k$, because otherwise assigning each agent one root from $R^*$ (as it would also be done by the greedy algorithm) already gives a partial allocation with optimal minimum satisfaction.

Let~$k':=k-|R^*|$.
Note that in any allocation, at least $k'$ agents do not receive any roots from~$R^*$.
In any optimal allocation, the total satisfaction remaining without the allocation of roots in $R^*$ has to suffice to give at least those $k'$  agents a satisfaction of at least~$S^*$. 
Hence, we get the following:
\begin{equation}\label{eq:fulfill}
    \sum_{i=1}^\ell (2a_i+1)+a -
    \sum_{i\in R^*} (a_i +1) 
    \geq k'\cdot S^*\,. 
\end{equation}

Now consider the partial greedy allocation $\pi'$ produced by the greedy algorithm after having assigned all the roots.
Let $A'$ denote the set of agents for which the satisfaction with respect to $\pi'$ is less than $S^*$, i.e., $ A'=\{ j \in K\colon  s_{\pi'}(j) < S^*\}$. 
We have $|A'|= k-|R^*|=k'$, since the greedy algorithm gives $|R^*|$ agents satisfaction at least $S^*$ by assigning them the roots in $R^*$, while all other agents receive at most one root by the assumption of this first case, and thus fail to reach $S^*$ in $\pi'$.

We argue that our greedy algorithm yields an optimal solution by showing that the number of unassigned leaves and singletons suffices to give all agents in $A'$ satisfaction of at least~$S^*$. 

To do so, we calculate, for all agents in $A'$, the total number $c$ of leaves and singletons required to reach $S^*$ by subtracting the satisfaction reached through the allocated root vertices from the desired target satisfaction.
We obtain:
\begin{eqnarray}
    c&=&\sum_{j \in A'} \left(S^*-\sum_{r_i \in \pi'(j)} (a_i +1)\right)\nonumber\\
    &=&|A'|\cdot S^*-\sum_{j \in A'} \sum_{r_i \in \pi'(j)} (a_i +1)\nonumber\\
    &=& k'\cdot S^*-\sum_{j \in A'} \sum_{r_i \in \pi'(j)} (a_i +1)\nonumber\\
    &\leq& 
    \sum_{i=1}^\ell (2a_i+1) + a -
     \sum_{i \in R^*} (a_i +1) 
  - \sum_{j \in A'}\sum_{r_i \in \pi'(j)} (a_i +1)
     \label{eq:plugin}\\ 
       &=&  \sum_{i=1}^\ell a_i +a \label{eq:allroots}
\end{eqnarray}
where (\ref{eq:plugin}) follows from plugging in (\ref{eq:fulfill}) and (\ref{eq:allroots}) is given by the fact that every root is either in $R^*$ or it was allocated to an agent in~$A'$.

But this upper bound is exactly the number of leaves and singletons that are available to be distributed by the greedy algorithm. 
By the exchange argument these leaves and singletons can be evenly distributed in a way that the satisfaction of each agent reaches $S^*$ in the allocation $\pi$ generated by the greedy algorithm.

\medskip
\noindent\textbf{Case 2.} \emph{At least one agent receives more than one root by the greedy algorithm.}
Note that this implies $\ell >k$.   
For each $i\in \{1,\ldots, \ell\}$, we denote by $\pi_i$ the partial greedy allocation after assigning the first $i$ roots, let $K_i$ be the set of agents who were assigned in $\pi_i$ at least two roots, let $\underaccent{\bar}{s}_i = \min\{ s_{\pi_i}(j) \colon j \in K  \}$ and let $\bar{s}_i = \max\{ s_{\pi_i}(j) \colon j \in K_i  \}$.

We first show that for all $i\in \{k+1,\ldots, \ell\}$, we have 
\begin{equation}\label{eq:bars-underbars}
   \bar{s}_i - \underaccent{\bar}{s}_i \leq a_{k+1} + 1 
\end{equation} holds. 
We prove this by induction on~$i$.
For $i = k+1$, when assigning the $(k+1)$-st root in the greedy algorithm, we may assume without loss of generality that this root is assigned to the agent that received the $k$-th root and in this partial assignment we have $\bar{s}_{k+1} = (a_k +1) + (a_{k+1}+1)$ and $\underaccent{\bar}{s}_{k+1} \geq a_k+1$. 
Hence $\bar{s}_{k+1} - \underaccent{\bar}{s}_{k+1} \leq a_{k+1} + 1$.
Now let $k+1\le i<\ell$ and assume that $\bar{s}_i - \underaccent{\bar}{s}_i \leq a_{k+1} + 1$.
Note that after every further root is assigned, the minimum satisfaction value never decreases, that is, $\underaccent{\bar}{s}_{i+1}\ge \underaccent{\bar}{s}_i$. 
When assigning the root $r_{i+1}$ (recall that $a_{i+1}\leq a_{k+1}$), the root is assigned to an agent with satisfaction at least $\underaccent{\bar}{s}_i$.
Note that $K_i\subseteq K_{i+1}$ and therefore $\bar{s}_i \le \bar{s}_{i+1}$.
If $\bar{s}_{i+1} = \bar{s}_{i}$, then using the induction hypothesis 
we immediately obtain that  $\bar{s}_{i+1} - \underaccent{\bar}{s}_{i+1} \leq a_{k+1} + 1$. 
If $\bar{s}_{i+1}>\bar{s}_{i}$, then the least satisfied agent with respect to $\pi_i$ becomes the most satisfied agent in $K_{i+1}$ and its satisfaction with respect to $\pi_{i+1}$ is $\bar{s}_{i+1} = \underaccent{\bar}{s}_{i} + a_{i+1} + 1$.
Therefore, \[\bar{s}_{i+1} - \underaccent{\bar}{s}_{i+1} =
\underaccent{\bar}{s}_{i} + a_{i+1} + 1-\underaccent{\bar}{s}_{i+1} \le a_{i+1} + 1 \leq a_{k+1} + 1\,.\]
This proves \eqref{eq:bars-underbars}.
In particular, the partial assignment $\pi_\ell$ after assigning all roots satisfies $\bar{s}_\ell - \underaccent{\bar}{s}_\ell \leq a_{k+1} + 1$.

Let $K'= K\setminus K_\ell$ be the set of agents that are assigned exactly one root in $\pi_\ell$.
Let $K'' = \{j \in K' \colon s_{\pi_\ell}(j) \leq \bar{s}_\ell \}$.
Observe that the preference graph, since $\ell>k$, contains at least  $(k+1)\cdot a_{k+1}$ leaves. 
Using the fact that $k+1>|K|\ge |K''\cup K_\ell|$ and~\eqref{eq:bars-underbars} we infer that 
\[(k+1)\cdot a_{k+1} >(|K''\cup K_\ell|) (\bar{s}_\ell-\underaccent{\bar}{s}_\ell-1)\,.\]
These leaves are then distributed greedily to the agents, starting with the agents in $K'' \cup K_\ell$. 
Since the satisfaction of each of these agents is at least $\underaccent{\bar}{s}_\ell$, the total number of leaves suffices for lifting the satisfaction of every agent in $K'' \cup K_\ell$ to $\bar{s}_\ell-1$.
Since in the procedure above we showed that these leaves can be assigned so that the least satisfied agent receives an additional leaf in every step, in the end the difference in satisfaction between any two agents in $K''\cup K_\ell$ is at most~$1$.

Hence, our greedy strategy leads to an allocation $\pi$ which, when restricted to the agents in $K''\cup K_\ell$, is an optimal allocation for the reduced instance resulting from removing the roots assigned to some agents in $K'\setminus K''$. 
But then $\pi$ must also be optimal for our original instance, since each agent in $K' \setminus K''$ receives only one root, and has satisfaction at least as high as the minimum satisfaction among the agents in $K''\cup K_\ell$.
\end{proof}

It might be noted that the greedy strategy with exchanges stated above for the general case $k \geq 3$ does not work for the special case~$k=2$.
This can be shown by means of a simple example with~$a_1=10, a_2=a_3=a_4=1,a=0$. 
The general greedy strategy assigns root 1 to agent 1 and all the other roots to agent 2, resulting in satisfactions 16 and 14.
An optimal solution would give roots 1 and 2 to agent 1 and roots 3 and 4 to agent 2 resulting in a satisfaction of 15 for both agents.

\medskip
For a constant number of agents, we can also give a polynomial-time algorithm for general out-forests.
This generalizes considerably the result for out-stars given in Theorem~\ref{th:minmax-out-stars}, but at the cost of restricting~$k$.

\begin{theorem}\label{th:same-tree-const-MINMAX}
For any constant number $k$ of agents, {\MINMAX} can be solved in polynomial time if the preference graph $G$ is an out-forest.
\end{theorem}

\begin{proof}
We denote by $n$ the number of vertices of~$G$.
For a vertex $v$, we denote by $T_v$ the out-tree rooted at $v$, that is, the subgraph of $G$ induced by $v$ and all its successors.
A \textit{dissatisfaction profile} is a $k$-tuple $(d_i\colon  i\in K)$ with $d_i \in \mathbb N_0$ for all $i\in K$ such that there is an allocation $\pi$ with  $\diss_{\pi}(i)=d_i$ for each $i\in K$.
Let ${\it DP}[v]$ be the set of all dissatisfaction profiles for~$T_v$.
Since a dissatisfaction profile is a vector of length $k$ whose elements have values between $0$ and $n$, the set ${\it DP}[v]$ contains at most $(n+1)^k$ profiles.

The \emph{depth} of a vertex $v$ is the number of vertices on the unique maximal directed path in $G$ ending at~$v$.
We observe that one can preprocess the input out-forest $G$ by only considering the vertices at depth at most $k$ in each out-tree. 
To verify this, observe that as a consequence of~\Cref{structure-of-optimal-solutions} we can assume that each agent receives at most one item from each path in $G$ in an optimal allocation $\pi$. We can iteratively strictly decrease the number of vertices with maximum depth $>k$  that are assigned to some agent (if such vertices exist) without making any agent worse off: take such a vertex $w$ assigned to some agent $i$ and the path $p$ from the root to~$w$. On $p$, by the above property, there must be at least one vertex $v\not=w$  that is not assigned to any agent; now, modify $\pi$ by allocating $v$ to agent $i$ and removing all items of $\text{succ}(v)$ from the set $\pi(i)$. 

Furthermore, note that if a vertex is not assigned to any agent in some solution, then assigning this vertex to any agent would also yield a solution.

Fix a vertex $u$ of~$G$.
Let $v_1,\dots,v_q$ be the out-neighbors of $u$, let $\pi_u$ be an arbitrary allocation for $T_u$ and let $d_u$ be the dissatisfaction profile of $\pi_u$.
We may assume that $u$ is assigned to some agent $j$ with respect to $\pi_u$, that is, $u \in \pi_u(j)$, and that no other vertex in $T_u$ is assigned to agent $j$; otherwise, we can remove this item from $\pi_u(j)$ without changing the dissatisfaction of any agent.
The dissatisfaction for agent $j$ with respect to $\pi_u$ is $0$, since $u \in \pi(j)$ and $u$ is the most preferred item in~$T_u$.
Moreover, for every agent $j' \neq j$, its dissatisfaction with respect to $\pi_u$ corresponds to the sum of the agent's dissatisfactions over all of the subtrees $T_{v_i}$ plus $1$, since $u \notin \pi_u(j')$.
Hence, we obtain that $d_u[j] = 0$ and, for all $j' \in K$ such that $j' \neq j$, $d_u[j'] = 1+ \sum_{i = 1}^q d_{v_i}[j']$, where $d_{v_i}$ corresponds to the dissatisfaction profile of the out-tree $T_{v_i}$ with respect to $\pi_u$.

We claim that ${\it DP}[u]$, the set of all dissatisfaction profiles for $T_u$, can be computed in polynomial time assuming that ${\it DP}[v_i]$ is known for all out-neighbors $v_i$ of~$u$.
The procedure is as follows.
If $u$ is a leaf, that is, $u$ has no out-neighbor, then ${\it DP}[u]$ is the set containing exactly $k+1$ profiles such that each profile is a vector of length $k$ with at most one element equal to $0$ and the other elements are equal to $1$; this assumes that $d \geq 1$, otherwise either $k = 1$ and the problem becomes trivial or there is no solution.
On the other hand, if $u$ is not a leaf, then let $S = \{0\}^k$, that is, $S$ contains a unique vector of length $k$ all of whose elements are~$0$.
Then, for each out-neighbor $v_i$ of $u$, we add to $S$ the profiles $d_S + d_i$ for all $d_S \in S$ and all $d_i \in {\it DP}[v_i]$.
We can detect duplicate profiles in $S$ in constant time, for instance by using a Boolean array $B$ with $(n+1)^k$ cells, where each cell corresponds to a dissatisfaction profile,
such that a dissatisfaction profile $d_i$ belongs to $S$ if and only if the cell corresponding to profile $d_i$ in $B$ is set to~$true$.
Hence, each time a new profile is computed, we can detect and ignore duplicates in $S$, which implies that at every step of the procedure, $|S| \leq (n+1)^k$.
For each neighbor $v_i$, since $|S|$ and $|{\it DP}[v_i]|$ have size at most $(n+1)^k$, it takes $\mathcal{O}(|S| \cdot |{\it DP}[v_i]|) = \mathcal{O}((n+1)^{2k})$ time to update~$S$.
Hence, once all neighbors of $u$ have been considered, it takes at most $\mathcal{O}({\it deg}^+(u)(n+1)^{2k})$ time to compute $S$, where ${\it deg}^+(u)$ denotes the out-degree of~$u$.
Note that, at the end of the procedure, $S$ corresponds exactly to the dissatisfaction profiles of the out-forest obtained from $T_u$ by removing~$u$.
Now, we explain how to construct the set ${\it DP}[u]$ of all dissatisfaction profiles for~$T_u$.
We iterate over $j \in K$ and for each $d_S \in S$, we create a profile $d_u$ such that $d_u[\ell] = d_S[\ell]+1$ if $\ell \neq j$, and~$d_u[j] = 0$.
Clearly, $d_u$ is a dissatisfaction profile for $T_u$ corresponding to an allocation $\pi_u$ such that $u \in \pi_u(j)$.
At the end of this loop, ${\it DP}[u]$ contains exactly all dissatisfaction profiles for~$T_u$.
Thus, if ${\it DP}[v_i]$ is known for all out-neighbors $v_i$ of $u$, then ${\it DP}[u]$ can be obtained in time $\mathcal{O}((n+1)^{2k+1})$.

Using a dynamic programming approach, we can compute the set ${\it DP}[v]$ for all vertices $v$ of $G$ in time $\mathcal{O}(n(n+1)^{2k+1}) = \mathcal{O}((n+1)^{2k+2})$.
Then, to decide whether there exists a solution for $G$, it suffices to compute the set of all dissatisfaction profiles for $G$ with a similar approach as the one we described to compute ${\it DP}[u]$ for a fixed $u$ and check the existence of a profile whose values are all at most~$d$.
Thus, {\MINMAX} can be solved in $\mathcal{O}((n+1)^{2k+2})$ time.
\end{proof}

The above theorem follows also as a special case of the general result for graphs of bounded treewidth stated in Theorem~17 and Remark~18 in \cite{general}.
However, for self-containment we gave an explicit construction of an elementary dynamic programming algorithm.

\section{Fixed parameter tractability based on modular partitions}\label{sec:modular}

The following result modifies the graph concept of bounded neighborhood diversity introduced by Lampis~\cite{Lampis12}.
In a (directed) graph $G=(V,A)$ a set $X \subseteq V$ is a \emph{module} if each vertex in $V \setminus X$ has a homogeneous relationship to $X$, i.e., for all $x, y \in X$ it holds that $N^-(x) \cap (V \setminus X) = N^-(y) \cap (V \setminus X)$ and $N^+(x) \cap (V \setminus X) = N^+(y) \cap (V \setminus X)$.
A partition of the vertex set $V$ of a graph $G = (V,E)$ into modules $X_1, \dots, X_d$, i.e., the $X_i$ are pairwise disjoint and $ \bigcup_{1 \leq i \leq d} X_i = V$, is called a \emph{modular partition} of~$G$. 
We call $d$ the \emph{size} of the modular partition. Typically, one restricts the types of induced subgraphs $G[X_i]$ to obtain modular partitions into certain module types.
One well-studied example is partitioning the vertex set of an undirected graph into clique modules and independent set modules, as in the definition of the neighborhood diversity.
We look at a partition into \emph{path modules} and independent set modules, hence the subsets $X_i$ in the definition above are either independent sets or directed paths.
Clearly, cliques are not very meaningful in the setting of preferences and partial orders. 
However, ordered paths are a decisive element of expressing preferences and fit well into the concept of modular partition.

The following theorem shows that such a partition is unique and can be computed efficiently.

\begin{theorem}\label{th:module-decomp}
A partition of a graph $G=(V,A)$ into a minimum number of path modules and independent set modules is unique and can be computed in $\mathcal{O}(|V|+|A|)$ time.
\end{theorem}

\begin{proof}
We use the unique modular decomposition tree of the graph $G=(V,A)$, which can be obtained in linear time using the algorithm of McConnell and Montgolfier~\cite{mcconnell2005linear}. 
Observe that both path modules and independent set modules uniquely decompose into singletons in the modular decomposition.
(Clearly, no proper subset of at least two path vertices, adjacent or not, can be a module to the other vertices of the ordered path.)
Hence, using the modular decomposition tree, 
we obtain our partition into path modules and independent set modules by iterating over all 
leaves of the modular decomposition tree which correspond to the singleton sets. 
If the parent module of such a singleton is a path module or an independent set module, we add
this module to our partition. 
Otherwise, we add the singleton module to the partition.
The uniqueness and minimality of this partition follows from the uniqueness of the modular decomposition tree and the above observation that path modules and independent set modules only decompose into singleton modules.
\end{proof}

\begin{theorem}\label{thm:nd}
{\MINMAX} is fixed parameter tractable with respect to $k+d$, where $k$ is the number of agents and $d$ is the minimum number of path modules and independent set modules partitioning~$V(G)$.
\end{theorem}

\begin{proof}
After computing the partition of $V(G)$ according to Theorem~\ref{th:module-decomp} we first consider the path modules. 
Since every agent will receive at most one vertex from every path module (by \Cref{structure-of-optimal-solutions}), we can limit all paths to have at most $k$ vertices and ignore additional vertices.
For paths with fewer than $k$ vertices we add dummy vertices whose allocation to an agent represents the fact that no vertex of this path is allocated to this agent. 
For each path module we {\em guess} one of the $k\,!$ orderings of the agents. 
Each such guess represents the ordering in which all the $k$ path vertices are assigned to the agents in a top-down way (dummy vertices mean that no vertex of that path is assigned to the corresponding agent).
Considering all guesses over all at most $d$ path modules yields $\mathcal{O}((k!)^d)$ possibilities for the allocation of path vertices.

Next, we consider the $d' \leq d$ independent set modules and {\em guess} for each of them and for each agent whether that agent is assigned at least one vertex of that module.
Clearly, there are at most $(2^k)^d$ possible guesses.

For each independent set module $M_i$, $i=1, \ldots, d'$, let $k(M_i)\leq k$ denote the number of agents receiving at least one vertex by the current guess.
Each of these $k(M_i)$ agents is allocated an arbitrary single vertex of $M_i$, if this is possible. 
By the module property, the actual choice is irrelevant.

Thus, together with the allocation of path vertices, every guess results in a partial allocation which may also include assignments of dominated vertices.
Since we restrict our computation to antichains (cf.~Lemma~\ref{structure-of-optimal-solutions}), we remove all dominated vertices to obtain a valid candidate consisting of a partial allocation with a satisfaction value $\sigma_j$ for every agent $j\in K$.
It remains to distribute the remaining vertices from all independent set modules in order to maximize the minimum satisfaction value.
This objective value is determined by iterating over all possible target satisfaction values $\sigma$.
Note that an allocation by guessing is out of the question since the size of a module $M_i$ can be $\Theta(|V(G)|)$.
Instead, the completion of the partial allocation is done by means of the following auxiliary maximum network flow model.

The network starts with a source node $s$ connected to all nodes $a_1, \ldots, a_k$ with $a_j$ representing agent~$j$. 
Each arc $(s, a_j)$ has capacity $\sigma-\sigma_j$ (agents with satisfaction $\sigma_j \geq \sigma$ are omitted in the flow model).
For every independent set module $M_i$, $i=1, \ldots, d'$, we add a node $b_i$ connected to a sink node $t$ where the capacity of the arc $(b_i,t)$ equals the number of unassigned vertices in $M_i$, i.e.,~$|M_i|-k(M_i)$.
Every agent node $a_j$ is connected (with infinite capacity) to every module node $b_i$ if $j$ already receives vertices from $M_i$ in the current guess (which implies that additional vertices from $M_i$ could be assigned to $j$).
If the maximum $s,t$-flow in this network has a value equal to $\sum_{j=1}^k (\sigma-\sigma_j)$, i.e., all arcs leaving the source $s$ are saturated, then the flow value on each arc $(a_j, b_i)$ corresponds to the number of (arbitrarily chosen) vertices from the independent set $M_i$ allocated to agent~$j$.
The resulting allocation has satisfaction value $\sigma$ for each agent.
Iterating over $\sigma$ (or performing binary search) one can find the best possible satisfaction reachable by the current guess.
Evaluating all guesses will solve {\MINMAX}.
\end{proof}

With our current techniques it does not seem to be possible to get rid of the exponential 
dependence on $k$ in the running time of the algorithm.
In fact, the complexity of {\MINMAX} is open even in the case when $G$ consists of three paths, with no connections between them (see Section~\ref{sec:conc}).
But if we consider a partition into independent set modules only, we are able to obtain an FPT algorithm for an arbitrary number of agents. 
Note that such a partition can be computed in linear time using the same techniques as used in the proof of \Cref{th:module-decomp}.
In the case of undirected graphs, maximal independent set modules were studied by Heggernes, Meister, and Papadopoulos~\cite{MR2857693} and by Lozin and Milani{\v c}~\cite{MR2658491} (under the name \emph{similarity classes}).

\begin{theorem}
For any number $k$ of agents, {\MINMAX} is fixed parameter tractable with respect to $d$, where $d$ is the minimum number of independent set modules partitioning~$V(G)$.
\end{theorem}

\begin{proof}
Let $\mathcal{I}$ be the family of independent sets obtained by the partition of $V(G)$ into the minimum number of independent set modules.
A nonempty set $S \subseteq \mathcal{I}$ is said to be  \emph{assignable} if it holds for all pairs $I, I' \in S$, $I\neq I'$, that there is no directed path in $G$ from a vertex in $I$ to a vertex in $I'$; in other words, if the set $\bigcup_{I \in S} I$ is a nonempty antichain in~$G$.
We denote by $\mathcal{A} \subseteq 2^{\mathcal{I}}$ the set of assignable sets.
Note that $|\mathcal{A}| \leq 2^{d}$.
For an assignable set $S \in \mathcal{A}$ we call an agent an \emph{$S$-agent} if it is only assigned vertices from $\bigcup_{I \in S} I$ and for each independent set $I \in S$ at least one vertex from $I$ is assigned to the agent. 
Recall that by \Cref{structure-of-optimal-solutions} we may restrict our attention to allocations in which the set of vertices allocated to each agent is a nonempty antichain.
For any such allocation, each agent is an $S$-agent where $S$ is the set of independent sets containing vertices assigned to the agent, and is not an $S'$-agent for any set $S'\in \mathcal{A}\setminus\{S\}$.

As a first step we guess a subset $\mathcal{A}' \subseteq \mathcal{A}$ of assignable sets.
This guess corresponds to the case when for each $S \in \mathcal{A}'$ at least one $S$-agent exists and for all $S \in \mathcal{A} \setminus \mathcal{A}'$ no $S$-agent exists.
We then iterate over all dissatisfaction thresholds $t$ and check if a feasible allocation to all $k$ agents exists that is consistent with the guess $\mathcal{A}'$ and gives a dissatisfaction at most $\delta$ for all agents.
Clearly, for the optimal dissatisfaction value also, such an allocation must exist for some guess of assignable sets.
We introduce an integer programming model that performs exactly such a check for a given $\mathcal{A}'$ and $\delta$.

\medskip
For every $S \in \mathcal{A}'$ let $d_{S}$ be the number of vertices of $G$ not dominated by any 
$v \in \bigcup_{I \in S} I$, i.e., the guaranteed dissatisfaction of every $S$-agent.
We introduce integer variables $x_{S}$ for each $S \in \mathcal{A}'$ that count how many agents are $S$-agents. 
In addition, we introduce integer variables $y_{S,I}$ for each $S \in \mathcal{A}'$ and $I \in S$ that count how many vertices of $I$ are assigned to $S$-agents.
Expressed in terms of these variables, the dissatisfaction of every $S$-agent is at most
\begin{equation}\label{eq:bound_proof}
        d_S + \left|\bigcup_{I \in S} I \right| - \left\lfloor \sum_{I \in S} \frac{y_{S, I}}{x_S} \right\rfloor ,
\end{equation}
since the $\sum_{I \in S} y_{S,I}$ vertices out of the independent sets in $S$ assigned to the $S$-agents can be distributed evenly among them.
The upper bound in (\ref{eq:bound_proof}) must be smaller or equal to the given threshold value $\delta$.
Note that by integrality of $\delta$
the floor function in (\ref{eq:bound_proof}) can be dropped in this condition, yielding
\begin{equation}\label{eq:t-bound_proof}
 \delta \geq d_S + \left|\bigcup_{I \in S} I \right| - \sum_{I \in S} \frac{y_{S, I}}{x_S}\,. 
\end{equation}
We multiply (\ref{eq:t-bound_proof}) by $x_{S} \geq 1$ to obtain the following equivalent linear inequality:
\begin{equation*}
     \delta\, x_{S} \geq d_S\, x_S + \left|\bigcup_{I \in S} I \right| x_S - \sum_{I \in S} y_{S, I}
\end{equation*}
Hence, the following integer programming decision formulation checks if the dissatisfaction threshold $t$ can be reached by an allocation to $k$ agents that is consistent with our guess $\mathcal{A}'$.
\begin{align}
     \delta\, x_{S} &\geq d_S\, x_S + \left|\bigcup_{I \in S} I \right| x_S - \sum_{I \in S} y_{S, I} & \forall S \in \mathcal{A}'\label{eq:t-bound-mult_proof}\\
     \sum_{S\in \mathcal{A}'} x_S &= k & \\
\sum_{S \in \mathcal{A}'\colon\, I \in S} y_{S,I} &\leq |I| & \forall I \in \mathcal{I}\\
     x_{S} &\geq 1 & \forall S \in \mathcal{A}'\\
     y_{S,I} &\geq x_{S} & \forall S \in \mathcal{A}', I \in S\\
    x_{S}, y_{S,I} &\in \mathbb{N} & \forall S \in \mathcal{A}', I \in S\label{eq:binary}
\end{align}
Note that the nonzero coefficients of this integer program are either constant ($\pm 1$) or polynomial in the input size (namely $\delta$, $d_S$, $\left|\bigcup_{I \in S} I \right|$, and $k$), and the number of variables and constraints is bounded by a function of~$d$.
Hence, using Lenstra's algorithm~\cite{lenstra1983integer} the solution of this integer program is fixed parameter tractable with respect to parameter~$d$.

Taking over all guesses $\mathcal{A}'$ the minimum value of $\delta$ that allows a feasible solution of (\ref{eq:t-bound-mult_proof})--(\ref{eq:binary}) gives an allocation with optimal dissatisfaction value. Thus, {\MINMAX} is fixed parameter tractable with respect to~$d$.
\end{proof}

\section{Conclusions}
\label{sec:conc}

In this paper we study the fair allocation of items to agents, when preferences are expressed by a mutually agreed preference graph. 
Fairness is represented by the maximum dissatisfaction experienced by any agent.
The setting with individual preferences was treated in a related paper~\cite{general}.
For general graphs, {\MINMAX} is efficiently solvable when we are dealing with $2$ agents and turns out to be \NP-complete for $k\geq 3$ agents.
However, for $k\geq 3$ agents, when we restrict the preference graph to out-stars, the problem can be solved in polynomial time.
Polynomiality can be shown also for out-forests, if the number of agents $k$ is a constant.
For graphs allowing certain bounded modular partitions, fixed parameter tractability results can be derived.

\medskip
As an interesting open problem for {\MINMAX} and an arbitrary number of agents, we suggest the case where $G$ consists only of out-paths.
W.l.o.g.\ paths can be restricted to have $k$ vertices. 
Note that the problem is trivial for shorter paths and in the case of two out-paths, where we can just assign the items to agents in increasing, respectively decreasing, order on the two paths to reach an optimal solution for the optimization version of {\MINMAX}.

For the special case of three out-paths each of $k$ vertices, an optimal solution can be given as shown in Figure~\ref{fig:three-paths}.
Note that one has to distinguish the cases of $k$ being even or odd, and that the allocation described there satisfies the lower bound $3 \cdot (k-1)/2$.
However, {\MINMAX} remains open for three out-paths of arbitrary length.

\begin{figure}
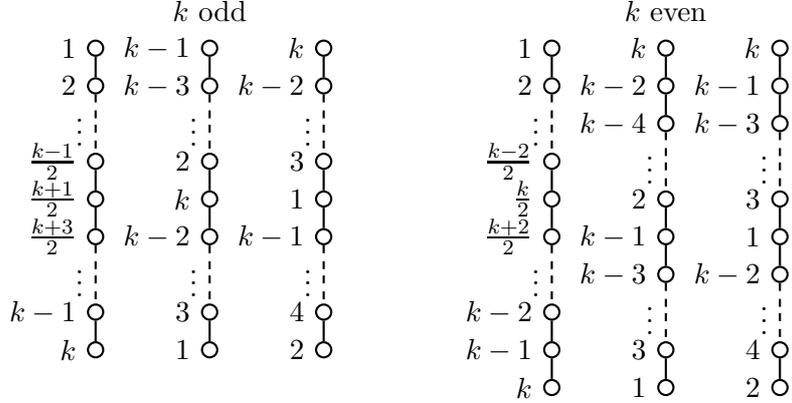

    \centering
\picturepathsamelength
\caption{If paths have length (at least) $k$, then the natural lower bound $3(k-1)/2$ obtained by 
summing up all dissatisfaction values over all vertices, which gives $3k(k-1)/2$, and dividing this value by $k$,
can be obtained by the above construction based on the parity of~$k$.
\label{fig:three-paths}}
\end{figure}

Considering a more difficult problem, it would be interesting to establish the complexity of {\MINMAX} for general out-forests with an arbitrary number of agents~$k$. 
This would close the gap between Theorem~\ref{th:minmax-out-stars} (out-stars for arbitrary $k$) and Theorem~\ref{th:same-tree-const-MINMAX} (out-forest for constant $k$).

Recall that {\MINMAX} can be solved in polynomial time if the graph has width~$2$ (see \Cref{th:width-two-MINMAX-MINSUM}).
In a more general direction, we point to the question of the complexity of {\MINMAX} on preference graphs of bounded width, which were recently studied in a number of papers~\cite{MR4459341,MR4402362,MR4364288,MR3473350,MR3951102}. 
Note that if both the width $w$ and the number of agents $k$ is bounded, it is easy to obtain in linear time a kernel of size $k\cdot w$ by reducing the preference graph to the top $k$ levels, implying fixed parameter tractability in~$k+w$.

\medskip
A different objective would consider the sum of dissatisfaction values over all agents, instead of the maximum.
The resulting minimization follows a utilitarian approach and maximizes the overall efficiency of the allocation.
This setting was already introduced in~\cite{adt2021} and will be subject of our future work.

\subsubsection*{Acknowledgements.}

The work of this paper was done in the framework of a bilateral project between University of Graz and University of Primorska, financed by the OeAD (SI 13/2023) and the Slovenian Research and Innovation Agency (BI-AT/23-24-009).
The authors acknowledge partial support by the Slovenian Research and Innovation Agency (I0-0035, research programs P1-0285 and P1-0404, and research projects N1-0102, N1-0160, N1-0210, J1-3001, J1-3002, J1-3003, J1-4008, and J1-4084), by the research program CogniCom (0013103) at the University of Primorska, and by the Field of Excellence ``COLIBRI'' at the University of Graz.


\bibliographystyle{plain}
\bibliography{refs}

\begin{thebibliography}{10}

\bibitem{ALOYSIUS2006273}
J.A. Aloysius, F.D. Davis, D.D. Wilson, A.~{Ross Taylor}, and J.E. Kottemann.
\newblock User acceptance of multi-criteria decision support systems: The
  impact of preference elicitation techniques.
\newblock {\em European Journal of Operational Research}, 169(1):273--285,
  2006.

\bibitem{aziz20}
H.~Aziz.
\newblock Developments in multi-agent fair allocation.
\newblock In {\em Proceedings of the 34th AAAI Conference on Artificial
  Intelligence (AAAI'20)}, pages 13563--13568. AAAI, 2020.

\bibitem{aziz}
H.~Aziz, S.~Gaspers, S.~Mackenzie, and T.~Walsh.
\newblock Fair assignment of indivisible objects under ordinal preferences.
\newblock {\em Artificial Intelligence}, 227:71--92, 2015.

\bibitem{MR4459341}
J.~Balab\'{a}n and P.~Hlin\v{e}n\'{y}.
\newblock Twin-width is linear in the poset width.
\newblock In {\em 16th {I}nternational {S}ymposium on {P}arameterized and
  {E}xact {C}omputation}, volume 214 of {\em LIPIcs. Leibniz Int. Proc.
  Inform.}, pages Art. No. 6, 13. Schloss Dagstuhl. Leibniz-Zent. Inform.,
  2021.

\bibitem{BANAECOSTA1994489}
C.A. {Bana e Costa} and J.-C. Vansnick.
\newblock {MACBETH} -- an interactive path towards the construction of cardinal
  value functions.
\newblock {\em International Transactions in Operational Research},
  1(4):489--500, 1994.

\bibitem{baumeister2017}
D.~Baumeister, S.~Bouveret, J.~Lang, T.~Nguyen, J.~Rothe, and A.~Saffidine.
\newblock Positional scoring-based allocation of indivisible goods.
\newblock {\em Autonomous Agents and Multi-Agent Systems}, 31:628--655, 2017.

\bibitem{bezakova2005}
I.~Bez\'{a}kov\'{a} and V.~Dani.
\newblock Allocating indivisible goods.
\newblock {\em ACM SIGecom Exchanges}, 5(3):11--18, 2005.

\bibitem{MR4402362}
\'{E}. Bonnet, E.J. Kim, S.~Thomass\'{e}, and R.~Watrigant.
\newblock Twin-width {I}: {T}ractable {FO} model checking.
\newblock {\em Journal of the ACM}, 69(1):Art. 3, 46, 2022.

\bibitem{bouveret-survey}
S.~Bouveret, Y.~Chevaleyre, and N.~Maudet.
\newblock Fair allocation of indivisible goods.
\newblock In F.~{Brandt et al.}, editor, {\em Handbook of Computational Social
  Choice}, chapter~12. Cambridge University Press, 2016.

\bibitem{bouveretlang2008}
S.~Bouveret and J.~Lang.
\newblock Efficiency and envy-freeness in fair division of indivisible goods:
  Logical representation and complexity.
\newblock {\em Journal of Artificial Intelligence Research}, 32:525--564, 2013.

\bibitem{brams2000}
S.J. Brams and P.C. Fishburn.
\newblock Fair division of indivisible items between two people with identical
  preferences: Envy-freeness, {Pareto}-optimality, and equity.
\newblock {\em Social Choice and Welfare}, 17:247--267, 2000.

\bibitem{assign2012}
R.~Burkard, M.~Dell'Amico, and S.~Martello.
\newblock {\em Assignment Problems}.
\newblock Society for Industrial and Applied Mathematics, 2012.

\bibitem{MR4364288}
M.~C\'{a}ceres, M.~Cairo, B.~Mumey, R.~Rizzi, and A.I. Tomescu.
\newblock A linear-time parameterized algorithm for computing the width of a
  {DAG}.
\newblock In {\em Graph-theoretic concepts in computer science}, volume 12911
  of {\em LNCS}, pages 257--269. Springer, 2021.

\bibitem{adt2021}
N.~Chiarelli, C.~Dallard, A.~Darmann, S.~Lendl, M.~Milani{\v{c}},
  P.~Mur{\v{s}}i{\v{c}}, U.~Pferschy, and N.~Piva{\v{c}}.
\newblock Allocating indivisible items with minimum dissatisfaction on
  preference graphs.
\newblock In {\em Proceedings of the 7th International Conference on
  Algorithmic Decision Theory (ADT'21)}, volume 13023 of {\em LNCS}, pages
  243--257. Springer, 2021.

\bibitem{general}
N.~Chiarelli, C.~Dallard, A.~Darmann, S.~Lendl, M.~Milani{\v{c}},
  P.~Mur{\v{s}}i{\v{c}}, U.~Pferschy, and N.~Piva{\v{c}}.
\newblock Allocation of indivisible items with individual preference graphs.
\newblock {\em Discrete Applied Mathematics}, 334:45--62, 2023.

\bibitem{cornilly2022}
D.~Cornilly, G.~Puccetti, L.~Rüschendorf, and S.~Vanduffel.
\newblock Fair allocation of indivisible goods with minimum inequality or
  minimum envy.
\newblock {\em European Journal of Operational Research}, 297(2):741--752,
  2022.

\bibitem{darmann2009}
A.~Darmann, C.~Klamler, and U.~Pferschy.
\newblock Maximizing the minimum voter satisfaction on spanning trees.
\newblock {\em Mathematical Social Sciences}, 58(2):238--250, 2009.

\bibitem{Amico97}
M.~Dell'Amico and S.~Martello.
\newblock The $k$-cardinality assignment problem.
\newblock {\em Discrete Applied Mathematics}, 76(1-3):103--121, 1997.

\bibitem{Dilworth}
R.P. Dilworth.
\newblock A decomposition theorem for partially ordered sets.
\newblock {\em Annals of Mathematics. Second Series}, 51:161--166, 1950.

\bibitem{freeman}
R.~Freeman, E.~Micha, and N.~Shah.
\newblock Two-sided matching meets fair division.
\newblock In {\em Proceedings of the 30th International Joint Conference on
  Artificial Intelligence, {IJCAI'21}}, pages 203--209, 2021.

\bibitem{MR0078334}
D.R. Fulkerson.
\newblock Note on {D}ilworth's decomposition theorem for partially ordered
  sets.
\newblock {\em Proceedings of the American Mathematical Society}, 7:701--702,
  1956.

\bibitem{MR3473350}
J.~Gajarsk\'{y}, P.~Hlin\v{e}n\'{y}, D.~Lokshtanov, J.~Obdr\v{z}\'{a}lek,
  S.~Ordyniak, M.S. Ramanujan, and S.~Saurabh.
\newblock F{O} model checking on posets of bounded width.
\newblock In {\em 2015 {IEEE} 56th {A}nnual {S}ymposium on {F}oundations of
  {C}omputer {S}cience---{FOCS} 2015}, pages 963--974. IEEE Computer Society,
  2015.

\bibitem{gareyjohnstock76}
M.R. Garey, D.S. Johnson, and L.~Stockmeyer.
\newblock Some simplified {NP}-complete graph problems.
\newblock {\em Theoretical Computer Science}, 1(3):237--267, 1976.

\bibitem{Golovin05}
D.~Golovin.
\newblock Max-min fair allocation of indivisible goods.
\newblock Technical report, School of Computer Science, Carnegie Mellon
  University, 2005.

\bibitem{hamm91}
P.J. Hammond.
\newblock {\em Interpersonal comparisons of utility: Why and how they are and
  should be made}, pages 200--254.
\newblock Cambridge University Press, 1991.

\bibitem{MR2857693}
P.~Heggernes, D.~Meister, and C.~Papadopoulos.
\newblock Graphs of linear clique-width at most 3.
\newblock {\em Theoretical Computer Science}, 412(39):5466--5486, 2011.

\bibitem{MR0337699}
J.E. Hopcroft and R.M. Karp.
\newblock An {$n^{5/2}$} algorithm for maximum matchings in bipartite graphs.
\newblock {\em SIAM Journal on Computing}, 2:225--231, 1973.

\bibitem{kawase22}
Y.~Kawase and H.~Sumita.
\newblock Online max-min fair allocation.
\newblock In {\em Algorithmic game theory}, volume 13584 of {\em Lecture Notes
  in Comput. Sci.}, pages 526--543. Springer, Cham, 2022.

\bibitem{kilgour2018}
D.M. Kilgour and R.~Vetschera.
\newblock Two-player fair division of indivisible items: Comparison of
  algorithms.
\newblock {\em European Journal of Operational Research}, 271(2):620--631,
  2018.

\bibitem{Lampis12}
M.~Lampis.
\newblock Algorithmic meta-theorems for restrictions of treewidth.
\newblock {\em Algorithmica}, 64:19--37, 2012.

\bibitem{lenstra1983integer}
H.W. Lenstra~Jr.
\newblock Integer programming with a fixed number of variables.
\newblock {\em Mathematics of Operations Research}, 8(4):538--548, 1983.

\bibitem{lipton04}
R.J. Lipton, E.~Markakis, E.~Mossel, and A.~Saberi.
\newblock On approximately fair allocations of indivisible goods.
\newblock In {\em EC'04: Proceedings of the 5th ACM conference on Electronic
  commerce}, pages 125--131. ACM, 2004.

\bibitem{MR2658491}
V.~Lozin and M.~Milani\v{c}.
\newblock On the maximum independent set problem in subclasses of planar
  graphs.
\newblock {\em Journal of Graph Algorithms and Applications}, 14(2):269--286,
  2010.

\bibitem{MR3951102}
V.~M\"{a}kinen, A.I. Tomescu, A.~Kuosmanen, T.~Paavilainen, T.~Gagie, and
  R.~Chikhi.
\newblock Sparse dynamic programming on {DAG}s with small width.
\newblock {\em ACM Transactions on Algorithms}, 15(2):Art.~29, 1--21, 2019.

\bibitem{mcconnell2005linear}
R.M. McConnell and F.~De~Montgolfier.
\newblock Linear-time modular decomposition of directed graphs.
\newblock {\em Discrete Applied Mathematics}, 145(2):198--209, 2005.

\bibitem{nguyen}
N.~Nguyen, T.~Nguyen, M.~Roos, and J.~Rothe.
\newblock Computational complexity and approximability of social welfare
  optimization in multiagent resource allocation.
\newblock {\em Autonomous Agents and Multi-Agent Systems}, 28(2):256--289,
  2014.

\bibitem{MR545530}
S.C. Ntafos and S.L. Hakimi.
\newblock On path cover problems in digraphs and applications to program
  testing.
\newblock {\em IEEE Transactions on Software Engineering}, 5(5):520--529, 1979.

\bibitem{proc14}
A.D. Procaccia and J.~Wang.
\newblock Fair enough: guaranteeing approximate maximin shares.
\newblock In {\em Proceedings of the 15th AAAI Conference on Artificial
  Intelligence (AAAI'14)}, pages 675--692. AAAI, 2014.

\bibitem{Rawls1971}
J.~Rawls.
\newblock {\em A theory of justice}.
\newblock The Belknap Press of Harvard University Press, Cambridge,
  Massachusetts, 1971.

\bibitem{roosarothe}
M.~Roos and J.~Rothe.
\newblock Complexity of social welfare optimization in multiagent resource
  allocation.
\newblock In {\em Proceedings of the 9th International Conference on Autonomous
  Agents and Multiagent Systems (AAMAS'10)}, pages 641--648, 2010.

\bibitem{thomson}
W.~Thomson.
\newblock Introduction to the theory of fair allocation.
\newblock In F.~{Brandt et al.}, editor, {\em Handbook of Computational Social
  Choice}, chapter~11, pages 261--283. Cambridge University Press, 2016.

\end{thebibliography}

\end{document}